\documentclass[11pt]{article}
\usepackage{rpmacros}
\usepackage[dvipsnames,svgnames,x11names]{xcolor}
\usepackage{amsthm}
\usepackage{xfrac}
\usepackage[T1]{fontenc}
\usepackage{enumerate}
\usepackage{verbatim}
\usepackage{booktabs}

\allowdisplaybreaks

\usepackage{thmtools}
\usepackage[colorlinks,pagebackref]{hyperref}
\hypersetup{
  linkcolor=[rgb]{0, 0, 0.6},
  citecolor=[rgb]{0, 0.4, 0},
  urlcolor=[rgb]{0.6, 0, 0},
  linktocpage=true,
  hypertexnames=false
}
\usepackage{mathrsfs,bbm}
\usepackage{todonotes}
\usepackage{tikz}
\usetikzlibrary{decorations.pathreplacing,backgrounds}
\usepackage{multirow}
\usepackage{setspace}
\usepackage[ruled, linesnumbered, vlined]{algorithm2e}
\SetKwComment{Comment}{/* }{ */}
\SetKwFunction{Prune}{Prune}
\SetKwFunction{FastDESolver}{FastDESolver}
\SetKwFunction{FastFESolver}{FastFESolver}
\SetKwFunction{ShortVector}{ShortVector}
\SetKwFunction{ConstructBasis}{ConstructBasis}
\SetKwFunction{UMultListRecovery}{UMultListRecovery}
\let\oldnl\nl
\newcommand{\nonl}{\renewcommand{\nl}{\let\nl\oldnl}}

\usepackage{amsmath}

\renewcommand{\epsilon}{\varepsilon}
\newcommand{\eps}{\varepsilon}
\renewcommand{\tilde}{\widetilde}
\newcommand{\Enc}{\mathsf{Enc}}

\newif\ifdraft
\draftfalse
\ifdraft
\else
\fi

\newcommand{\ehref}[1]{\href{mailto:#1}{#1}}

\newcommand{\parhead}[1]{\medskip \noindent {\bfseries \ignorespaces {#1}}\hskip 0.9em}

\newcommand{\smcdef}{\sum_{j=0}^{\ell}c_j \alpha^j}
\newcommand{\smudef}{\sum_{j=0}^{\ell}u_j \alpha^j}
\newcommand{\ignore}[1]{}

\DeclareMathOperator{\spn}{span}

\renewcommand{\le}{\leqslant}
\renewcommand{\leq}{\leqslant}
\renewcommand{\ge}{\geqslant}
\renewcommand{\geq}{\geqslant}

\usepackage[nameinlink,capitalise]{cleveref}
\crefname{lemma}{Lemma}{Lemmas}
\crefname{theorem}{Theorem}{Theorems}
\crefname{definition}{Definition}{Definitions}
\crefname{claim}{Claim}{Claims}
\crefname{table}{Table}{Tables}

\newtheorem{theorem}{Theorem}[section]
\newtheorem{lemma}[theorem]{Lemma}

\newtheorem{corollary}[theorem]{Corollary}

\theoremstyle{definition}

\newtheorem{definition}[theorem]{Definition}
\newtheorem{remark}[theorem]{Remark}

\newcommand{\out}{\mathrm{out}}
\newcommand{\Cout}{\calC_{\mathrm{out}}}
\newcommand{\Eout}{\Enc_{\mathrm{out}}}
\newcommand{\Eoutj}[1]{\Enc_{\mathrm{out},#1}}
\newcommand{\Cin}{\calC_{\mathrm{in}}}
\newcommand{\Ein}{\Enc_{\mathrm{in}}}
\newcommand{\Einl}[1]{\Enc_{\mathrm{in},#1}}
\newcommand{\kin}{k_{\mathrm{in}}}
\newcommand{\CAEL}{\calC_{\mathrm{AEL}}}

\newcommand{\EAELi}[1]{\Enc_{\mathrm{AEL},#1}}

\newcommand{\Vleft}{V_{\mathrm{left}}}
\newcommand{\Vright}{V_{\mathrm{right}}}

\newcommand{\wt}{\widetilde}

\usepackage{nth}
\usepackage{intcalc}
\usepackage{etoolbox}
\usepackage{xstring}

\usepackage{ifpdf}
\ifpdf
\else
\usepackage[quadpoints=false]{hypdvips}
\fi

\newcommand{\ECCCurl}[2]{http://eccc.hpi-web.de/report/\ifnumcomp{#1}{>}{93}{19}{20}#1/#2/}

\newcommand{\shortECCC}[2]{\texttt{\href{http://eccc.hpi-web.de/report/\ifnumcomp{#1}{>}{93}{19}{20}#1/#2/}{eccc:TR#1-#2}}}

\newcommand{\parseECCC}[1]{
\StrSubstitute{#1}{TR}{}[\tmpstring]%
\IfSubStr{\tmpstring}{/}{ 
\StrBefore{\tmpstring}{/}[\ecccyear]%
\StrBehind{\tmpstring}{/}[\ecccreport]%
}{
\StrBefore{\tmpstring}{-}[\ecccyear]%
\StrBehind{\tmpstring}{-}[\ecccreport]%
}%
\shortECCC{\ecccyear}{\ecccreport}}

\usepackage[letterpaper,margin=1.1in]{geometry}
\title{Explicit Constant-Alphabet Subspace Design Codes} 

\author{
{Rohan Goyal\thanks{Massachusetts Institute of Technology, Cambridge \ehref{rohan\_g@mit.edu}.}}
\and 
{Venkatesan Guruswami\thanks{University of California, Berkeley \ehref{venkatg@berkeley.edu}.}  } \and 
{Jun-Ting Hsieh \thanks{Massachusetts Institute of Technology, Cambridge \ehref{juntingh@mit.edu}.}}
}

\date{April 16, 2026}

\begin{document}

\maketitle

\thispagestyle{empty}

\begin{abstract}
The subspace design property for additive codes is a higher-dimensional generalization of the minimum distance property.
As shown recently \cite{brakensiek2025random}, it implies that the code has similar performance as random linear codes with respect to all ``local properties''. 
Explicit algebraic codes, such as folded Reed-Solomon and multiplicity codes, are known to have the subspace design property, but they need alphabet sizes that grow as a large polynomial in the block length.
Constructing explicit constant-alphabet subspace design codes was subsequently posed as an open question in \cite{brakensiek2025random}.

In this work, we answer their question and give explicit constructions of subspace design codes over constant-sized alphabets, using the expander-based Alon-Edmonds-Luby (AEL) framework. This generalizes the recent work of \cite{jeronimo2025AEL}, which showed that such codes share local properties of random linear codes.
Our work obtains this consequence in a unified manner via the subspace design property. In addition, our approach yields some improvements in parameters for list-recovery.

\end{abstract}
\newpage

\tableofcontents
\thispagestyle{empty}

\newpage

\section{Introduction}

In coding theory, a fundamental goal is to understand the trade-offs between different parameters of a code, such as its rate, distance, alphabet size, as well as error-tolerance properties like list-decodability and list-recoverability.
In most parameter regimes, the best known trade-offs are obtained via probabilistic arguments showing that some random ensemble of codes meet the desired bounds.
These existential results provide a benchmark for what is information-theoretically possible.

For many applications, however, one often seeks explicit deterministic constructions that approach the guarantees of random codes.
Algebraic constructions provide a natural rich family of codes with strong parameters.
These include Reed-Solomon codes, folded Reed-Solomon codes~\cite{GR08}, and multiplicity codes~\cite{GW13,KSY14,Kopparty2014}.
Recently, there has been significant progress on the list size for list-decoding these codes~\cite{KoppartyRSW2023, Tamo2024, Srivastava2025,ChenZ2025,AHS25,brakensiek2025random}.
In particular, it is now known~\cite{ChenZ2025} that these codes achieve an optimal trade-off of list-decodability up to an error fraction $1-R-\eps$ with list size $O(1/\eps)$, where $R$ is the rate.
However, a major limitation is that these well-studied algebraic codes typically require the alphabet size to be growing with the block length of the code.\footnote{An exception is the family of Algebraic-Geometry (AG) codes.
There are explicit constructions of optimal-rate list-decodable AG codes over constant-sized alphabets \cite{GuoR22,GX22}, but analogous results for list-recovery are still unknown.}

This leads to an important question: can one construct explicit codes over a \emph{constant-sized alphabet} (independent of the block length) with similar guarantees as random ones?
Recently, there has been significant progress in this direction, particularly through combinatorial constructions based on expander graphs.
A remarkable work~\cite{JeronimoMST2025} showed how expander-based constructions via the ``AEL framework'' (which is also the focus in this work and will be described in \Cref{sec:AEL-overview}) can yield constant-alphabet codes matching the list-decoding performance of algebraic codes such as folded Reed-Solomon codes (including the optimal list size).
A subsequent work~\cite{SrivastavaT2025} even gave near-linear-time list-decoding algorithms with such guarantees.
Following this paradigm, a recent work of Jeronimo and Shagrithaya~\cite{jeronimo2025AEL} showed that the AEL framework can transform (small) random linear codes into explicit codes while preserving a broad class of ``local'' properties.

\parhead{Local properties.}
We first give a brief overview of \emph{local coordinate-wise linear (LCL) properties} introduced in \cite{LMS-focs25}, which generalize various previous notions of local properties \cite{ShangguanT20,MRRSW20,GM22,GuruswamiMRSW22,GZ23,GM-DA24}.
We will restrict to $\F$-additive codes (i.e., $\F$-linear codes with alphabet $\F^s$ for some $s\in \N$).
Informally, a property of a code is LCL if it is witnessed by a small set of distinct codewords $x^{(1)},\dots,x^{(L)} \in \calC$ such that, for each $i\in[n]$, the entries $(x_i^{(j)})_{j\in [L]}$ satisfy certain linear constraints.
Typically, we are interested in the \emph{complement} of LCL properties; that is, we would like a code that does not contain distinct codewords $x^{(1)}, \dots,x^{(L)}$ satisfying any property in some family of LCL properties.
For example, a code is list-decodable if it does \emph{not} contain $L$ codewords close to any arbitrary vector, and this can be precisely captured by such local linear constraints.

The LCL framework captures several fundamental properties including list-decodability, list-recoverability, average pairwise distance~\cite{CheraghchiGV2013}, list-decodability from burst errors \cite{RothV09}, and proximity gaps \cite{GGproximity25}.
With this unified lens, Levi, Mosheiff and Shagrithaya \cite{LMS-focs25} showed precise threshold theorems with respect to LCL properties for random linear codes in $\F_q$ as well as random Reed-Solomon codes, hence recovering several known results \cite{BGM22,AGGLZ25}.

Motivated by this, in an exciting work, Jeronimo and Shagrithaya \cite{jeronimo2025AEL} showed explicit constructions that achieve the same guarantees as random linear codes for \emph{any} LCL property.
Their construction is based on the Alon-Edmonds-Luby (AEL) construction \cite{AEL95}, which is a generalization of the graph-based distance amplification technique of \cite{ABNNR92}.
This framework has since been used in several explicit code constructions \cite{KMRS17,HW18,KoppartyRSW2023,JeronimoMST2025,SrivastavaT2025,jeronimo2025AEL}.

An underlying theme of the AEL framework is a ``local-to-global'' phenomenon,\footnote{The local-to-global principle has been highly successful across theoretical computer science, from the classical Tanner codes \cite{Tanner81,SS96}, the PCP theorem \cite{AS98,ALMSS98,Dinur2007}, to recent developments such as locally testable codes \cite{DELLM22,PK22} and vertex expanders \cite{HMMP24,HLMOZ25,HLMRZ25}.}
where one uses a ``pseudorandom'' object, such as an expanding graph, to lift properties of a small constant-sized object to an infinite family of large objects.
Jeronimo and Shagrithaya showed that this paradigm extends to LCL properties.
Starting from a constant-sized code with LCL properties matching random linear codes (which can be obtained by brute force), the AEL construction gives an infinite family of codes that inherit these properties.
We define the AEL construction and provide a more detailed discussion of \cite{jeronimo2025AEL} in \Cref{sec:AEL-overview,sec:JS-overview}.

\parhead{Subspace design captures all local properties.}
In another novel recent work, Brakensiek, Chen, Dhar, and Zhang \cite{brakensiek2025random}
showed that a property known as \emph{subspace design} simultaneously captures \emph{all} LCL properties.
A subspace design, first introduced by Guruswami and Xing \cite{GX-stoc13}, is a collection $\{H_1,\dots,H_n\}$ of linear subspaces in $\F^m$ such that for any low-dimensional subspace $\calA \subseteq \F^m$, the average intersection $\E_{i}\dim(H_i \cap \calA)$ is much smaller than $\dim(\calA)$.
For an $\F_q$-additive code $\calC \subseteq (\F_q^s)^n$, we say that it is a subspace design code if the collection $\{\calC_i\}_{i\in [n]}$, where $\calC_i = \{x\in \calC \mid x_i=0\}$, forms a subspace design; see \Cref{def: Subspace-designs} for the formal definition.
Note that this property necessarily requires the code to be \emph{folded}, i.e., over an alphabet $\F_q^s$ for some sufficiently large $s$.
Indeed, if $\calC$ is $\F$-linear over $\F$ and $\calA \subseteq \calC$ is any subspace, then the restriction $x_i=0$ imposes only a single linear constraint, and thus $\dim(\calC_i \cap \calA) \geq \dim(\calA) - 1$.

Guruswami and Kopparty \cite{GK16} showed that folded Reed–Solomon codes and univariate multiplicity codes satisfy the subspace design property. 
The aforementioned work of Chen and Zhang \cite{ChenZ2025} on optimal list-decodability of folded Reed-Solomon codes in fact showed that the subspace design property with optimal parameters  implies optimal list-decodability. Specifically, this implies that folded Reed-Solomon codes and univariate multiplicity codes achieve optimal list-decoding bounds.

Recently, Brakensiek et.~al.~\cite{brakensiek2025random} showed that being a subspace design code is the unifying property for the LCL framework of \cite{LMS-focs25}.
At a high level, they showed that if random codes with rate $R$ avoids a local property $\calV$ with high probability, then any near-optimal subspace design code 
with rate close to $R$ also does not contain $\calV$.
By \cite{GK16}, their result implies explicit codes that simultaneously simulate all local properties of random linear codes.

However, such codes require the alphabet size to be at least polynomial in the block length.
For constant-alphabet subspace design codes, \cite{brakensiek2025random} proved existence by showing that a (folded) random linear code satisfies the subspace design property with high probability (see \Cref{thm: Folded RLCs are subspace-designs}).
They further posed the explicit construction of near-optimal subspace design codes over constant-sized fields as an open question.

\parhead{Our work: explicit subspace design codes.}
We give explicit constructions of \emph{subspace design codes} over constant-sized alphabets, This resolves the above-mentioned open question of \cite{brakensiek2025random}, and also abstracts and generalizes the framework of \cite{JeronimoMST2025,jeronimo2025AEL}.
As mentioned earlier, the subspace design property alone implies (the avoidance of) \emph{all} LCL properties shown for random linear codes and the construction of \cite{jeronimo2025AEL}.

\begin{theorem} \label{thm:subspace-design}
    For any finite field $\FF_q$, positive integer $r$, and reals $R,\epsilon \in (0,1)$, there exists an explicit $\F_q$-additive code $C \subseteq \Sigma^n$ of rate $R$ and alphabet $\Sigma = \FF_q^{\poly(r,1/\epsilon)\cdot q^{r^2}}$ such that for any subspace $\calA \subseteq \calC$ of dimension at most $r$, we have that
    \begin{align} \label{eq:subspace-design}
        \frac{1}{n} \sum_{i=1}^n \dim (\calA_i)\le (R+\epsilon) \cdot \dim(\calA)  \,,
    \end{align}
    where $\calA_i = \{x\in \calA \mid x_i = 0\}$.
\end{theorem}

The guarantee in \Cref{thm:subspace-design} is precisely the definition of a subspace design code.
In words, for any low-dimensional subspace $\calA \subseteq \calC$, restricting a coordinate to $0$ reduces the dimension to roughly $R \cdot \dim(\calA)$ on average.
Our alphabet size is $\exp(\exp(r^2))$ when $q$ is a constant, and it is an interesting question whether a better analysis could yield a smaller alphabet.
We note that because of the $q^{r^2}$ term, we can even choose $\eps \leq q^{-o(r^2)}$.

\begin{remark}[Distance and the subspace design parameter]
\label{rem:subspace-design-parameter}
    Suppose $\calA = \spn\{x^*\}$ is a $1$-dimensional subspace, where $x^*\in \calC$ is a minimum-weight codeword of relative weight $\delta$.
    Then, $\frac{1}{n} \sum_{i=1}^n \dim(\calA_i) = 1 - \delta$, since $\calA_i = \{0\}$ whenever $x_i^* \neq 0$.
    Thus, \Cref{eq:subspace-design} implies that the distance of our code $\delta \geq 1-R-\eps$, which approaches the Singleton bound $\delta \leq 1- R + o(1)$.
    
    On the other hand, the Singleton bound also implies that $\frac{1}{n}\sum_{i=1}^n \dim(\calA_i) \geq R - o(1)$.
    In particular, \Cref{eq:subspace-design} in \Cref{thm:subspace-design} is almost optimal.
\end{remark}

Focusing on the single subspace design property allows us to obtain an arguably simpler analysis and better parameters for downstream applications.
To compare with prior work, we consider list-recovery, which was in fact one of the original motivations of studying LCL properties and subspace design.
Using the reduction from subspace design to list-recovery established in \cite{brakensiek2025random}, our construction achieves a list size bound that matches the best known bound for any linear code.
Moreover, among explicit constructions with near-optimal list sizes, our construction has the smallest alphabet size.
In \Cref{sec:comparison}, we give a more detailed discussion, including a comparison with prior work in \Cref{tab:comparison}.

For concreteness, we also include list-decoding, list-recovery, and curve-decoding parameters for our construction in \Cref{sec:consequences}.
These results follow directly from the results of \cite{ChenZ2025,brakensiek2025random,GGproximity25}.

\subsection{The Alon-Edmonds-Luby (AEL) construction}
\label{sec:AEL-overview}

Our construction is based on the Alon-Edmonds-Luby (AEL) framework \cite{AEL95}, which we present in this section.

\begin{definition}[The Alon-Edmonds-Luby (AEL) construction]
\label{def:AEL}
    Given
    \begin{enumerate}[(1)]
        \item an outer code $\Cout$ with an encoder $\Eout : \Sigma_{\out}^k \to (\Sigma^{\kin})^n$,
        \item an inner code $\Cin$ with an encoder $\Ein : \Sigma^{\kin} \to \Sigma^d$, and
        \item a $d$-regular bipartite graph $G = (\Vleft, \Vright, E)$ with $\Vleft = \Vright = [n]$, where for each vertex, there is a fixed ordering of its incident edges.
    \end{enumerate}
    The AEL code $\CAEL = \CAEL(\Cout, \Cin, G) \subseteq (\Sigma^d)^n$ is defined as follows.
    For any message $x\in \Sigma_{\out}^k$, we first encode it with the outer code and obtain $y = \Eout(x) \in (\Sigma^{\kin})^n$.
    Each entry $y_i$ is an element in $\Sigma^{\kin}$, and we encode it using the inner code, which gives $\Ein(y_i) \in \Sigma^{d}$.
    Then, we place each entry of $\Ein(y_i)$ on the $d$ incident edges of vertex $i$ in $\Vleft$.
    The final codeword $z \in (\Sigma^d)^n$ is such that for each $j\in \Vright$, $z_j \in \Sigma^d$ is the collection of elements on the incident edges of $j$.
\end{definition}

In most settings, $|\Sigma|$, $\kin$ and $d$ are constants, and the inner code is viewed as a \emph{random} code, since one can obtain a code with properties matching those of random codes by brute force.
The graph $G$ is chosen to be a spectral expander, for which explicit constructions are well known.
For the outer code $\Cout$, one may choose explicit codes tailored to the application.
In \cite{HW18}, list-recoverable codes (for erasures) were used, and it was shown that $\CAEL$ inherits this property.
In \cite{KMRS17}, they used multiplicity codes of \cite{KSY14} (in the sub-constant distance regime) and showed that $\CAEL$ inherits the local correctability of the outer code.
In \cite{KoppartyRSW2023}, they instantiated the outer code using folded Reed-Solomon codes and showed that $\CAEL$ inherits the list-recovery property.

In an exciting work, Jeronimo, Mittal, Srivastava, and Tulsiani \cite{JeronimoMST2025} showed that it suffices for the outer code to have constant distance in order for $\CAEL$ to be list-decodable up to capacity.
In particular, this removes any application-specific requirement on the outer code, yielding the first purely combinatorial codes approaching the generalized Singleton bound.
Their construction relies solely on properties of the ``random'' inner code, and they interpreted this as another instance of a ``local-to-global'' phenomenon, where properties of a constant-sized object transfer to an infinite family of large objects.

\subsection{Overview of our work}

Our construction also follows the ``local-to-global'' paradigm of the AEL framework (\Cref{def:AEL}).
It was proved in \cite{brakensiek2025random} that a folded random linear code satisfies near-optimal subspace design with high probability (\Cref{thm: Folded RLCs are subspace-designs}), and thus we can use it as the inner code.
In our main theorem (\Cref{thm:main-AEL-theorem}),
we show that the AEL construction inherits the subspace design property of the inner code, as long as the outer code has constant distance and the graph is a sufficiently good expander.

In our analysis, we use an abstract interpretation of the LCL properties in \cite{LMS-focs25}.
Previously, local properties are defined with respect to a locality parameter $L$ (e.g., the list size in the context of list-decoding and list-recovery).
In \cite{LMS-focs25}, a \emph{local profile} is defined to be a tuple of subspaces $(V_1,\dots,V_n)$, where each $V_i$ is a subspace in $\F_q^L$.
For technical issues related to differing field sizes and block lengths, \cite{jeronimo2025AEL} defined it as a tuple of matrices $(M_1,\dots,M_n)$ in $\F_q^{L\times L}$, and
a code $\calC$ contains a local profile if there are $L$ distinct codewords such that their $i$-th coordinates, viewed as a vector in $\F_q^L$, is in the kernel of $M_i$.

We define local profiles more abstractly as a tuple of subspaces $(V_1,\dots,V_n)$ in an arbitrary vector space $V$ over $\F_q$.
A $\F_q$-additive code $\calC \subseteq (\F_q^s)^n$ contains $(V_1,\dots,V_n)$ if there exist a subspace $A \subseteq \F_q^k$ (in the message space) and an isomorphism $\varphi : V \to A^\ast$, where $A^\ast$ denotes the dual space of $A$ (i.e., the set of linear functionals from $A \to \F_q$), that satisfy the following:
for each $i\in[n]$, $\varphi(V_i)$ (a subspace of $A^\ast$) has kernel (a subspace of $A$) contained in $\ker(\Enc_i)$, where $\Enc_i: \F_q^k \to \F_q^s$ is the encoding map for coordinate $i$.
See \Cref{def: local profiles} for the formal definition.
Although more abstract, we find that this definition avoids technicalities encountered in prior work and also leads to a more streamlined proof.

With this definition, we show an equivalence statement for subspace design codes (\Cref{thm: local-ish for subspace-design}):
the $\tau$-subspace design property is equivalent to the property that every local profile $(V_1,\dots,V_n)$ contained in $\calC$ satisfies $\Phi \geq 0$, where $\Phi$ is some potential function depending only on the parameter $\tau$ and the average of $\dim(V_i)$.
We note that similar potential functions were also used in \cite{LMS-focs25} and \cite{jeronimo2025AEL}.

With this equivalence in hand, our analysis roughly goes as follows.
Suppose our code contains any local profile $(V_1,\dots,V_n)$ with $\Phi < 0$, due to a subspace $A \subseteq \F_q^k$.
Then, transferring this to the left side of the AEL construction, the expansion of the graph implies that most inner codes also have local profiles $(V_{N_1(j)}, \dots, V_{N_d(j)})$ with negative potential; here $j \in \Vleft$ and $N_1(j),\dots,N_d(j)$ are its $d$ neighbors in $\Vright$.
However, by the subspace design property of the inner codes, this forces $\Eoutj{j}(A) = 0$ for most indices $j$, contradicting the distance of the outer code.

\subsection{The work of \texorpdfstring{\cite{jeronimo2025AEL}}{JS25}}
\label{sec:JS-overview}

In this section, we briefly discuss the work of \cite{jeronimo2025AEL}.
Jeronimo and Shagrithaya~\cite{jeronimo2025AEL} showed that any LCL property satisfied (or avoided) by the inner code can be transferred to the AEL code.
Since the inner code is a ``random'' code, we know precise characterizations of what LCL properties it avoids (due to \cite{LMS-focs25}).
Thus, if there are $L$ distinct codewords that satisfy a local profile, which they define as a tuple of $L \times L$ matrices $(M_1,\dots,M_n)$, then porting over to the left side, it follows that most inner codewords satisfy a close approximation of the local profiles, which would be a contradiction.

The above is a simplified overview of \cite{jeronimo2025AEL} that hides several details.
A challenge that they need to overcome is that the inner codewords may not be pairwise distinct.
To overcome this, they need to work with a \emph{robust} version of \emph{implied} local profiles (which also appear implicitly in \cite{LMS-focs25}), which we will not elaborate here.
One advantage of working directly with subspace design in our work is that we avoid such technical subtleties.

\parhead{Threshold rates may change over extension fields.}
A crucial detail we would like to point out is that the construction of \cite{jeronimo2025AEL}, particularly the inner code, is $\F_q$-linear, while the LCL properties are over a subfield $\F_{q_0}$ (i.e., $M_i \in \F_{q_0}^{L\times L}$ where $q_0 < q$ are powers of the same prime).

The inner code they chose is a random $\F_q$-linear code.\footnote{From personal communication, the authors of \cite{jeronimo2025AEL} suggest that their proof goes through if they chose the inner code to be $\F_{q_0}$-additive.}
This leads to their construction working only for local properties over $\FF_{q_0}$ whose threshold rates are the same for random linear codes over extension fields.  
It is conceivable that for some properties of interest, the threshold over an extension field may be different. This is \textit{not} a concern for list-decoding and list-recovery (the primary application of \cite{jeronimo2025AEL}), as we know that the thresholds for these properties are independent of the field extension.

On the other hand, one might consider some properties that $\FF_{q_0}$-additive codes can have, but for which it is unclear whether $\FF_q$-linear codes can have them.  As a concrete example, we mention a local property that was shown by \cite{GGproximity25} to hold for $\FF_{q_0}$-additive subspace design codes, but is not known to hold for codes linear over an extension of $\FF_{q_0}$. 

An $\FF_{q_0}$-additive code $C \subseteq \Sigma^n$ is said to be \textit{low-dimension-recoverable} if, for any $1$-dimensional $\FF_{q_0}$-linear spaces $A_1,\dots, A_n \subseteq \Sigma$, the space of all codewords $c\in A_1\times A_2 \times \cdots \times A_n$ is $O(1)$-dimensional.
This is naturally a property that linear codes over $\Sigma=\FF_{q_0}$ can never have, as then we would have $A_1\times\cdots\times A_n = \Sigma^n$. On the other hand, $\FF_{q_0}$-additive subspace design codes of rate $1/2$ 
do have this property \cite{GGproximity25}, but we do not know if rate $1/2$ $\FF_q$-linear codes for any proper extension $\F_q$ of $\F_{q_0}$ possess such low-dimensional recoverability.

Note that since we construct subspace design codes in this work, our codes naturally inherit \emph{all} LCL properties that hold for subspace design codes, including low-dimensional recoverability.

\subsection{Comparison to prior work on list-recovery}
\label{sec:comparison}

For comparison with prior work, we focus on list-recovery of our construction.
A code $\calC \subseteq \Sigma^n$ is $(\rho, \ell, L)$-list recoverable if for every collection of sets $S_1,S_2,\dots, S_n \subseteq \Sigma$, each $|S_i| \leq \ell$, there are at most $L$ codewords $c\in \calC$ that are distance $\leq \rho$ away from $S_1 \times S_2 \times \cdots \times S_n$, i.e.,
\begin{align*}
    \big| \big\{ c \in \calC \ \big| \  |\{ i\in[n]: c_i \notin S_i \}| \leq \rho n  \big\} \big| \leq L \,.
\end{align*}
Note that $(\rho, L)$-list-decodability is equivalent to $(\rho,1,L)$-list-recoverability.

Brakensiek et.~al.~\cite{BCZ25} showed a reduction from subspace design to list-recovery; see also \cite[Theorem 5.1]{GG25b}.
Plugging in \Cref{thm:subspace-design} to their reduction, we get the following as an immediate corollary.

\begin{corollary}[Consequence of \cite{BCZ25}]
    Our code $\calC \subseteq \Sigma^n$ from \Cref{thm:subspace-design} of rate $R$ and alphabet $\Sigma = \F_q^{\poly(\ell/\eps) \cdot q^{2\ell^2/\epsilon^2}}$
    is $(1-R-\eps, \ell, L)$-list-recoverable for
    \begin{align*}
        L \leq \left\lceil\inparen{ \frac{\ell}{R+\eps} }^{\frac{R}{\eps} + 1}\right\rceil \,.
    \end{align*}
\end{corollary}

\renewcommand{\arraystretch}{1.6}

\begin{table}[ht]
\centering
\begin{tabular}{|l|c|c|c|}
\hline
Constructions & List size $L$ & Alphabet size $|\Sigma|$ & Explicit? \\
\hline
Random linear codes \cite{LS25} & $(\ell/\eps)^{O(\ell/\eps)}$ & $\ell^{O(1/\eps)}$ & No \\
\hline
Random linear codes \cite{BCZ25} & $\inparen{ \frac{\ell}{R+\eps} }^{O(R/\eps+1)}$ & $\ell^{(\frac{\ell}{R+\eps})^{O(R/\eps)}/\eps}$ & No  \\
\hline \hline
\cite{KoppartyRSW2023} & $(\ell/\eps)^{(\ell/\eps)^2 \log(\ell/\eps)}$ & $\ell^{O(1/\eps^4)}$ & Yes \\
\hline
\cite{SrivastavaT2025} & $\exp\big( \exp \big( \wt{O}(\ell/\eps) \big) \big)$ & $\exp\big( \exp \big( \wt{O}(\ell/\eps) \big) \big)$ & Yes \\
\hline
\cite{brakensiek2025random} & $\inparen{ \frac{\ell}{R+\eps} }^{O(R/\eps+1)}$ & $(n\ell/\eps^2)^{O(\ell/\eps^2)}$ & Yes \\
\hline
\cite{jeronimo2025AEL} & $\inparen{ \frac{\ell}{R+\eps} }^{O(R/\eps+1)}$ & $\exp\big( \exp\big( \exp\big( \wt{O}(\ell/\eps) \big) \big) \big)$ & Yes \\
\hline
Our work & $\inparen{ \frac{\ell}{R+\eps} }^{O(R/\eps+1)}$ & $\exp \inparen{ \exp \inparen{ O(\ell^2/\eps^2) } }$ & Yes \\
\hline \hline
Any linear code \cite{LS25} & $\geq \ell^{\floor{\max(R/\eps,1)}}$ & Any & - \\
\hline
\end{tabular}
\caption{Adapted from \cite{KoppartyRSW2023,brakensiek2025random,
jeronimo2025AEL}. Constructions of $(1-R-\eps, \ell, L)$-list-recoverable codes of constant rate $R \in (0,1)$.}
\label{tab:comparison}
\end{table}

In \Cref{tab:comparison}, we compare our list-recovery parameters with prior constructions.
Note that our list size upper bound $L \leq \big( \frac{\ell}{R+\eps} \big)^{O(R/\eps+1)}$ matches the best known list size bound for random linear codes (in fact, any additive code), which is expected since optimal subspace design codes simulate all local properties of random linear codes \cite{brakensiek2025random}.
Moreover, this bound is almost tight as it was shown that all linear codes must have $L \geq \ell^{\Omega(R/\eps)}$ \cite{ChenZ2025,LMS-focs25,LS25}.
Among explicit constructions with near-optimal list size, folded Reed-Solomon and multiplicity codes \cite{brakensiek2025random} naturally require an alphabet size polynomial in $n$, while the construction of \cite{jeronimo2025AEL} has a constant alphabet but is exponentially larger than ours.

One concrete place where we get an improvement in the alphabet size is the following.
To study list-recovery using the LCL framework, one needs to look at high locality properties, namely $L = \ell^{O(1/\eps)}$-local LCL properties.
In our case, we only need to consider $O(\ell/\eps)$-dimensional subspaces, which saves an exponential factor.

\subsection{Organization}

We organize our paper as follows.
In \Cref{sec:prelims}, we introduce some preliminaries.
In \Cref{sec:characterization}, we define the potential function and local profiles, and we show equivalence statements for subspace design codes.
Then, in \Cref{sec:construction}, we prove our main theorem (\Cref{thm:main-AEL-theorem}), which implies \Cref{thm:subspace-design}.
Finally, in \Cref{sec:consequences}, we present the list-decoding, list-recovery, and curve-decoding parameters for our construction.

\section{Preliminaries}
\label{sec:prelims}

We begin by introducing the basic coding-theoretic definitions we will be using.
For any two vectors $x, y$ in $\Sigma^n$ where $\Sigma$ is some alphabet, we define the fractional Hamming distance $\Delta(x, y) = \frac{1}{n} |\{i\in [n]: x_i \ne y_i\}|$ to be the fraction of coordinates where they differ.
For a set $S\subseteq \Sigma^n$, we define $\Delta(x, S) = \min_{y\in S} \Delta(x,y)$
to be the fractional distance of $x$ to its closest vector in $S$.
Throughout this paper, all distances are taken to be fractional unless stated otherwise.

The two fundamental quantities associated with a code are its rate and distance.
For a code $\calC \subseteq \Sigma^n$, we define its (relative) distance as $\delta(\calC) = \min_{x, y \in \calC, x\ne y} \Delta(x,y)$.
Moreover, its rate $R(\calC)$ is defined as $R(\calC)= \frac{1}{n} \log_{|\Sigma|}|\calC|$.

In this paper, we will focus on additive codes over a finite field, defined as follows:

\begin{definition}[Additive codes]
    Let $\FF$ be a finite field and let $\Sigma = \FF^s$ for some positive integer $s$. 
    A code $\calC \subseteq \Sigma^n$ is said to be $\FF$-additive (or just additive when the field $\FF$ is clear from context) if $\calC$ is an $\FF$-linear subspace of $\Sigma^n$. When $s=1$, the code is simply called a linear code.
\end{definition}

\subsection{Vector spaces and dual spaces}

We will need the following linear-algebraic notations and definitions.
Let $V$ be a finite-dimensional vector space over a field $\F$.
We define $\calL(V)$ to be the set of all linear subspaces of $V$.
For subspaces $U, W$, we write $U + W = \{u + w \mid u \in U,\ w \in W\}$.
Suppose $W \subseteq U$, then the quotient space $U / W \coloneqq \{u + W : u \in U\}$ consists of the set of cosets.

\begin{definition} \label{def:dual-space}
    Let $V$ be a vector space over a field $\FF$.
    \begin{itemize}        
        \item The \emph{dual space} $V^*$ is the vector space of all linear maps $f: V \to \FF$.
        
        \item If $A \subseteq V$ is a subspace, its \emph{annihilator} is
        \[
        A^\perp \coloneqq \{ f \in V^* \mid f(a) = 0 \ \forall a \in A \} \,.
        \]
        
        \item If $B \subseteq V^*$ is a subspace, its \emph{joint kernel} is
        \[
        B^\circ \coloneqq \{ v \in V \mid f(v) = 0 \ \forall f \in B \} \,.
        \]
    \end{itemize}
    \end{definition}

We note that the annihilator and kernel are dual concepts, i.e., $(A^\perp)^\circ = A$.
For a subspace $B \subseteq V^\ast$ and a vector $v\in V$, we will use $Bv = 0$ to mean that $f(v) = 0$ for all $f\in B$.
Thus, we can write $B^\circ = \{v\in V \mid Bv=0\}$.

\subsection{Subspace design codes}

We now review the concept of subspace design codes.
Subspace designs have been the driving force in many recent advances in coding theory, including \cite{Srivastava2025,ChenZ2025, AHS25, BCZ25, brakensiek2025random,GGproximity25,GG25b}.
In the following, we use the notation as established in \cite{GGproximity25}.

\begin{definition}[Subspace-Design Property]
\label{def: Subspace-designs}
    For any function $\tau: \mathbb N\rightarrow \R_{\le 1}$, an $\FF_q$-additive code $\calC \subseteq (\FF_q^s)^n$ is said to be a $\tau$-\textit{subspace design} code if for every $r\in \mathbb N$, and every $\F_q$-linear subspace $\calA$ of $\calC$ of dimension at most $r$,
    \[ \frac{1}{n} \sum\limits_{i=1}^n \dim(\calA_i) \le \dim(\calA)\cdot \tau(r) \,, \]
    where $\calA_i = \{x \in \calA \mid x_i=0\}$.
\end{definition}

We will often use \Cref{def: Subspace-designs} where we define the subspaces in the \textit{message space} $\F_q^k$.
In particular, suppose the code $\calC = \{(\Enc_1(x),\dots, \Enc_n(x)) \mid x\in \F_q^k\}$.
Let $A' \subseteq \F_q^k$ is a subspace of dimension $\le r$ and $A = \Enc(A') \subseteq \calC$, then $A_i$ is precisely the subspace $\Enc(A' \cap \ker(\Enc_i))$, and we have the following,
\begin{align*}
    \frac{1}{n} \sum_{i=1}^n \dim(A' \cap \ker(\Enc_i)) \leq \dim(A') \cdot \tau(r) \,.
\end{align*}

Note that the subspace design property necessarily requires the code to be \emph{folded}, otherwise if $\calC$ is $\F_q$-linear over $\F_q$, then we have $\dim(\calA_i) \geq \dim(\calA) - 1$ because $x_i = 0$ imposes only a single linear constraint.

\begin{remark}[Parameter $\tau$]
    We can assume that $\tau(r)$ is a non-decreasing function, and moreover, the statement trivially holds when $\tau(r) = 1$.
    For our construction as well as known subspace design codes, $\tau$ is a function such that $\tau(r)$ is small for all small $r$, and $\tau(r) = 1$ otherwise.
    
    On the other hand, from the discussion in \Cref{rem:subspace-design-parameter}, the Singleton bound implies that $\tau(r) \geq R-o(1)$, where $R$ is the rate of the code.
    Thus, we would like \emph{near-optimal} subspace design codes where $\tau(r) \leq R+\eps$ for all small $r$ (via a sufficiently large alphabet size depending on $r$).
\end{remark}

\paragraph{Folded Reed-Solomon codes and folded random linear codes.}
An $s$-folded Reed-Solomon code \cite{GR08} is an $\F_q$-additive code in $(\F_q^s)^n$ where $q > sn$.
Given evaluation points $\alpha_1,\dots,\alpha_n \in \F_q$ and a primitive element $\gamma \in \F_q$ where $\alpha_i \gamma^t \neq \alpha_j$ for all $i\neq j$ and $t < s$, the code is obtained by encoding a polynomial $f\in \FF_q[x]$ of degree at most $k$ such that the $i$-th entry consists of the tuple $(f(\alpha_i),f(\alpha_i \gamma), \dots, f(\alpha_i \gamma^{s-1})) \in \F_q^s$.

Guruswami and Kopparty \cite{GK16} proved that folded Reed-Solomon (as well as univariate multiplicity codes) are optimal subspace design codes.
This subspace design property has been at the heart of the recent progress on list-decoding and list-recovery of these codes.

\begin{theorem}[\cite{GK16}]\label{thm: GKS16 subspace design}
    The $s$-folded Reed-Solomon codes of rate $R$ are $\tau$-subspace-design codes for $\tau(r)= \frac{sR }{s-r+1}$ for all $1\le r\le s$ and $\tau(r)=1$ otherwise.
\end{theorem}

It is also known that folded random linear codes form subspace design codes as well. This result was proved in \cite{brakensiek2025random}.

\begin{theorem}[\cite{brakensiek2025random}]
\label{thm: Folded RLCs are subspace-designs}
    Let $\eps > 0$.
    Let $\Enc_1, \dots, \Enc_n: \FF_q^k \to \FF_q^s$ be uniformly random linear functions.
    Then, with probability at least $1-o(1)$, $\calC = \{(\Enc_1(x),\dots,\Enc_n(x)) \mid x \in \FF_q^k\}$ is a $\tau$-subspace design code, where $\tau(r) \leq R+\eps$ for all $r \le \eps s/4$, and $R = k/sn$ denotes the code rate.
\end{theorem}

\subsection{List-decoding, list-recovery and curve-decoding from subspace designs}

It was shown that list-decoding, list-recovery and curve-decoding follow from the subspace design property \cite{ChenZ2025,BCZ25,GGproximity25}.
In this section, we state these implications, following the language of \cite{GGproximity25,GG25b}.

\begin{theorem}[Average radius list-decoding of subspace design codes \cite{ChenZ2025, BCZ25}]\label{thm: list-decoding}
    Any $\tau$-subspace design additive code $\calC\subseteq (\FF_q^s)^n$ has that for any $y\in (\FF_q^s)^n$, and any distinct $c_1, \dots, c_r\in \calC$, we have that \[\sum_{i=1}^r \Delta(y, c_i)\ge (r-1)(1-\tau(r-1)) \,. \]
    In particular, there exists an $i\in [r]$ such that $\Delta(c_i, y)\ge \frac{r-1}{r}(1-\tau(r-1))$. Thus, \[\left|\left\{c\in \calC\mid \Delta(c_i,y)<\frac{r-1}{r}(1-\tau(r-1))\right\}\right|\le r-1 \,. \]
\end{theorem}

\begin{theorem}[List-recovery of subspace design codes \cite{BCZ25}]\label{thm: list-recovery}
    Let $\ell\in \mathbb N$ and $\epsilon>0$. For any $\tau$ subspace design additive code $\calC\subseteq (\FF_q^s)^n$, for any $L_1,\dots, L_n\subseteq (\FF_q^s)$ such that $|L_i|=\ell$ for all $i\in [n]$, we have that \[\left|\left\{c\in \calC \mid \Delta(c_i,L_1\times L_2\times \cdots\times L_n)<1-\tau(\lceil\ell/\epsilon\rceil)-\epsilon \right\}\right|\le \left(\frac{\ell}{\tau(\lceil\ell/\epsilon\rceil)+\epsilon}\right)^{(\tau(\lceil\ell/\epsilon\rceil)+\epsilon)/\epsilon} \,. \]
\end{theorem}

The above statement was originally proven in \cite{BCZ25}, and a proof for the list recovery in our language of $\tau$-subspace design codes was also presented in the appendix of \cite{GG25b}.

Additionally, we present results for curve decoding of our code constructions. Thus, we begin by defining curve-decodadibility as defined in \cite{GGproximity25}. They showed that the subspace design property of codes implies curve decoding and consequentially the notions of correlated agreement and proximity gaps. For simplicity, we refrain from discussing the consequences and demonstrate the primary conclusion of curve decoding.

\begin{definition}[Curve-decodability]\label{def: curve-decodability}
    An $\F_q$-additive code $\mathcal C\subseteq \Sigma^n$ is $(\ell, \delta, a, b)$ curve-decodable if for every $u_0, u_1,\dots, u_\ell\in \Sigma^n$,  all functions $f: \FF_q \to \calC$, whenever the set 
    \[ A=\Bigl\{\alpha \in \FF_q \mid \Delta\bigl(\smudef, f(\alpha) \bigr) \le  \delta\Bigr\}\]
    has at least $a$ elements,  there exist $c_0, c_1,\dots, c_\ell \in \mathcal C$ such that 
    \[\Bigl|\bigl\{\alpha \in A  \mid  f(\alpha)=\smcdef \bigr\}\Bigr| \ge b \ .  \]
    When $\ell=1$, we use the phrase \emph{line-decodable.}
\end{definition}

\begin{theorem}[Curve-decodability of subspace design codes \cite{GGproximity25}]\label{thm: curve-decoding}
    For arbitrary positive integers $r, \ell, a$ and any $\epsilon \ge \frac{\ell+1}{r}$, every $\tau$-subspace design code $\mathcal C\subseteq \Sigma^n$ is $(\ell, 1-\tau(r)-\epsilon, a, \frac{\epsilon}{r+\epsilon}\cdot a)$ curve-decodable. 
\end{theorem}

\subsection{Expander graphs}

We review the concept of spectral expander graphs.

\begin{definition}[Spectral expander]
\label{def:expander}
    A $d$-regular bipartite graph $G = (\Vleft, \Vright, E)$ with $|\Vleft| = |\Vright| = n$ is called a $\lambda$-expander if its normalized biadjacency matrix $A$, i.e., $A_{i,j} = 1/d$ if $(i,j)\in E$ and $0$ otherwise, has second singular value $\sigma_2(A) \leq \lambda$.
\end{definition}

Normally, a spectral expander is defined as a non-bipartite $d$-regular graph whose normalized adjacency matrix satisfies $\max(\lambda_2(A), |\lambda_{\min}(A)|) \leq \lambda$.
One can convert an $n$-vertex $\lambda$-expander into a bipartite $\lambda$-expander on $[n]\times [n]$ via its double cover.

It is known that the optimal value of $\lambda$ is $\frac{2\sqrt{d-1}}{d}$ by the Alon-Boppana bound, and graphs achieving this bound are called Ramanujan graphs.
Although we do not require such a strong bound on $\lambda$, there exist many explicit constructions of near-Ramanujan graphs \cite{LPS88,Mar88,Mor94,MOP20,Alon21}.

\begin{theorem}[\cite{Alon21}]
\label{thm:expander-construction}
    For every degree $d$, every $\eps$, and all sufficiently large $n \geq n_0(d,\eps)$ where $nd$ is even, there is an explicit construction of $n$-vertex $d$-regular $\lambda$-expanders with $\lambda \leq \frac{2\sqrt{d-1}}{d} + \eps$.
\end{theorem}

\section{Characterization of subspace designs}
\label{sec:characterization}

\subsection{Potential function and local profiles}

We first define a potential function given a tuple of subspaces.
We note that this was also used in \cite{LMS-focs25,jeronimo2025AEL}.

\begin{definition}[Potential function]\label{def:potential-function}
    Let $V$ be any finite-dimensional vector space over a field $\FF$, and let $n\in \mathbb N$.
    Define $\Phi_V: \mathcal L(V) \times (\mathcal L(V))^n \times \R \to \R$ as follows:
    \[\Phi_V(U, (V_1,\dots, V_n), \alpha)= \alpha \dim(U) - \frac{1}{n} \sum_{i=1}^n \inparen{ \dim(U) - \dim (U\cap V_i) } \,. \]
    We will drop the subscript $V$ if it is clear from context.
\end{definition}

We will need the following result on the potential function when we apply a surjective linear map $M: V \to V'$, which may have a kernel.

\begin{lemma}\label{thm: quotienting is good}
    Let $V, V'$ be vector spaces, let $U, V_1,\dots, V_n \in \mathcal L(V)$, let $M: V\to V'$ be a surjective linear map such that $\ker(M)=W \in \calL(V)$, and let $\alpha \in \R$. Then,
    \[\Phi_V(U+W, (V_1,\dots, V_n), \alpha) - \Phi_V(W, (V_1,\dots, V_n), \alpha) = \Phi_{V'}(M(U), (M(V_1), \dots, M(V_n)), \alpha) \,. \] 
\end{lemma}

\begin{proof}
    By \Cref{def:potential-function}, it suffices to prove that for each $V_i$,
    \[\dim ((U+W)\cap V_i) - \dim(W \cap V_i) = \dim (M(U) \cap M(V_i)) \,.\] 
    Let $X = (U+W) \cap V_i$, and consider the restriction $M|_X : X \to V'$.
    We have that the image $\im(M|_X) = M(X) = M(U+W) \cap M(V_i) = M(U) \cap M(V_i)$, since $M(W) = \{0\}$.
    The kernel $\ker(M|_X) = X \cap W$, which equals $W \cap V_i$ because $(U + W) \cap W = W$.
    Then, by the rank-nullity theorem, $\dim(X) = \dim \ker(M|_X) + \dim \im(M|_X) = \dim(W \cap V_i) + \dim(M(U) \cap M(V_i))$.
\end{proof}

Next, we define local profiles.

\begin{definition}[Local profiles] \label{def: local profiles}
    Let $\Enc_1, \dots, \Enc_n: \FF_q^k \to \FF_q^s$ be linear maps for some $s \in \N$, and $\calC = \{(\Enc_1(x), \dots, \Enc_n(x))\mid x\in \FF_q^k\}$ be the associated $\FF_q$-additive code.
    Throughout this paper, we will assume that $\calC$ has positive distance, i.e., $\Enc: \F_q^k \to (\F_q^s)^n$ is injective.
    
    Let $V$ be any $\FF_q$-linear space. A sequence $(V_1,\dots, V_n)\in \mathcal L(V)^n$ is called a $V$-local profile, or simply local profile when $V$ is clear from context.

    We say that $\calC$ contains a $V$-local profile $(V_1,\dots, V_n)$ if there exists a subspace $A\subseteq \mathcal \FF_q^k$ and an isomorphism $\varphi : V \to A^\ast$ such that, for each $i\in [n]$, we have $\varphi (V_i)^\circ \subseteq \ker(\Enc_i)$;
    that is, any vector $a\in A$ with $\varphi(V_i)(a) = 0$ satisfies $\Enc_i(a) = 0$.
\end{definition}

Before proceeding, we first provide some intuition for \Cref{def: local profiles} and its connection to subspace design codes and the potential function $\Phi$ (\Cref{def:potential-function}).
In \Cref{def: local profiles}, $A \subseteq \F_q^k$ is in the message space.
For each $i\in[n]$, $\varphi(V_i)$ is a subspace of $A^\ast$, so the kernel $\varphi(V_i)^\circ$ is a subspace of $A$.
By the rank-nullity theorem, we have
\begin{align*}
    \dim(\varphi(V_i)^\circ) = \dim(A) - \dim(\varphi(V_i)) = \dim(V) - \dim(V_i) \,.
\end{align*}
Here, we use that $\varphi$ is an isomorphism from $V$ to $A^\ast$.

Since $\varphi(V_i)^\circ \subseteq A$, the inclusion $\varphi (V_i)^\circ \subseteq \ker(\Enc_i)$ implies $\dim(\varphi (V_i)^\circ) \leq \dim(A \cap \ker(\Enc_i))$.
Suppose $\wt{A} = \Enc(A) \subseteq \calC$ and $\wt{A}_i = \{x \in \wt{A} \mid x_i = 0\}$ (tying back to the definition of subspace design codes in \Cref{def: Subspace-designs}),
then $\wt{A}_i = \Enc(A \cap \ker(\Enc_i))$, and
\begin{align*}
    \frac{1}{n} \sum_{i=1}^n \dim(\wt{A}_i) = \frac{1}{n} \sum_{i=1}^n \dim(A \cap \ker(\Enc_i))
    \geq \frac{1}{n} \sum_{i=1}^n \dim(\varphi (V_i)^\circ)
    = \frac{1}{n} \sum_{i=1}^n (\dim(V)-\dim(V_i)) \,.
\end{align*}
The right-hand-side appears in the potential function $\Phi_V(V, (V_1,\dots,V_n), \alpha)$.
In \Cref{sec:equivalence-of-subspace-designs}, we will crucially rely on this connection to show a characterization of subspace design codes.

A key advantage of this definition of local profile is its versatility.
In the following, we show that the property of containing a local profile is preserved under taking quotients.

\begin{lemma} \label{lem: quotienting}
    Let $V$ be a linear space and $(V_1,\dots, V_n)$ be a $V$-local profile, and let $W\subseteq V$.
    If a code $\calC=\{(\Enc_1(x),\dots, \Enc_n(x)) \mid x\in \FF_q^k\}$ contains $(V_1,\dots, V_n)$, then it also contains $\left((V_1+W)/W, \dots, (V_n+W)/W \right)\in \mathcal L(V/W)^n$.
\end{lemma}
\begin{proof}
    By \Cref{def: local profiles}, if $\calC$ contains $(V_1,\dots,  V_n)$, then there exists a subspace $A\subseteq \mathcal \FF_q^k$ and an isomorphism $\varphi: V\to A^\ast$ such that $\varphi(V_i)^\circ \subseteq \ker(\Enc_i)$. 

    $\varphi(W)$ is a subspace in $A^\ast$, and let $A' = \varphi(W)^\circ$ be its kernel, which is a subspace in $A$ of dimension $\dim(V) - \dim(W) = \dim(V / W)$.
    Define a linear map $M: A^\ast \to A'^\ast$ given by the restriction $M(f) \coloneqq f|_{A'}$, where $M$ takes a functional on $A$ and restricts it to $A'$.
    The kernel of $M$ is the set of functionals $f \in A^\ast$ such that $f(a) = 0$ for all $a\in A'$,
    i.e., $\ker(M) = (A')^\perp$, the annihilator of $A'$, which means that $\ker(M) = (\varphi(W)^\circ)^\perp = \varphi(W)$ (see \Cref{def:dual-space}).
    Then, since $\varphi$ is an isomorphism, $\ker(M \circ \varphi) = \varphi^{-1}(\ker(M)) = W$.

    The map $M \circ \varphi: V \to A'^\ast$ is surjective.
    By the isomorphism theorem, $A'^\ast \cong V/\ker(M\circ \varphi) = V/W$, and there is an isomorphism $\varphi': V/W \to A'^\ast$ that maps a coset $v+W$ to $M(\varphi(v))$.

    Now, it suffices to prove that $\varphi'((V_i+W)/W)^\circ \subseteq \ker(\Enc_i)$.
    By definition, $\varphi'((V_i+W)/W) = M(\varphi(V_i))$.
    Since $M$ only restricts to evaluating on $A'$ without changing their values, for any $a\in A'$ we have $M(\varphi(v))(a) = \varphi(v)(a)$.
    Thus, any $a\in A'$ with $M(\varphi(V_i))(a) = 0$ also implies $\varphi(V_i)(a) = 0$, and by the assumption $\varphi(V_i)^\circ \subseteq \ker(\Enc_i)$, we have $a\in \ker(\Enc_i)$ as desired.
\end{proof}

\subsection{Equivalence statements for subspace design codes}
\label{sec:equivalence-of-subspace-designs}

With the definitions of local profiles and potential function in hand, we next show an equivalence statement for subspace design codes.

\begin{theorem}\label{thm: local-ish for subspace-design}
    An additive code $\calC=\{(\Enc_1(x),\dots, \Enc_n(x))\mid x\in \FF_q^k\}\subseteq \left(\FF_q^s\right)^n$ is a $\tau$-subspace design code if and only if for any $r\in \mathbb N$, vector space $V$ of dimension at most $r$, and every $V$-local profile $(V_1, \dots, V_n)$ that $\calC$ contains, we have
    \[\Phi_V(V, (V_1,\dots, V_n), \tau(r)) \ge 0 \,. \]
\end{theorem}

\begin{proof}
    We first show that if the conclusion does not hold, then $\calC$ is not a $\tau$-subspace design code.
    In particular, there exists a positive integer $r$, a vector space $V$ of dimension $\le r$, and a $V$-local profile $(V_1,\dots, V_n)$ that $\calC$ contains, such that $\Phi_V(V, (V_1,\dots, V_n), \tau(r)) < 0$.

    Recall from \Cref{def: local profiles} that $\calC$ containing $(V_1,\dots, V_n)$ means that there exists a subspace $A \subseteq \FF_q^k$, and an isomorphism $\varphi: V \to A^\ast$, such that $\varphi(V_i)^\circ \subseteq \ker(\Enc_i)$.
    In fact, since $\varphi(V_i)^\circ$ is a subspace in $A$, we have $\varphi(V_i)^\circ \subseteq \ker(\Enc_i) \cap A$.
    By the rank-nullity theorem, we have $\dim(\varphi(V_i)^\circ) = \dim(A^\ast) - \dim(\varphi(V_i)) = \dim(V) - \dim(V_i)$.
    Thus,
    \begin{align*}
        \frac{1}{n} \sum_{i=1}^n \dim(\ker(\Enc_i) \cap A)
        &\geq \frac{1}{n} \sum_{i=1}^n \dim(\varphi(V_i)^\circ)
        = \frac{1}{n} \sum_{i=1}^n (\dim(V) - \dim(V_i)) \\
        &> \tau(r) \cdot \dim(V) \,,
    \end{align*}
    where the last inequality follows from $\Phi_V(V, (V_1,\dots, V_n), \tau(r)) < 0$.
    
    Now, consider the subspace $\wt{A} = \Enc(A) \subseteq \calC$ of dimension $\dim(\wt{A}) = \dim(V) \leq r$, and let $\wt{A}_i = \{a \in \wt{A} \mid a_i=0\}$ for $i\in[n]$.
    We have that $\dim(\wt{A}_i) = \dim(\ker(\Enc_i) \cap A)$, and $\frac{1}{n} \sum_{i=1}^n \dim(\wt{A}_i) > \tau(r) \cdot \dim (\tilde A)$, which shows that $\calC$ is not a $\tau$-subspace design code.
    
    For the other direction, if $\calC$ is not a $\tau$-subspace design code, then there exists a subspace $A \subseteq \F_q^k$ (in the message space) of dimension at most $r$ such that $\frac{1}{n}\sum_{i=1}^n \dim(A \cap \ker(\Enc_i)) > \tau(r)\dim(A)$.
 
    Let $V = A^\ast$, and for $i \in [n]$, let $V_i = (A \cap \ker(\Enc_i))^\perp$ (the annihilator of $A \cap \ker(\Enc_i)$).
    By definition, we have $V_i^\circ \subseteq \ker(\Enc_i)$.
    Moreover, since $\dim(V_i) = \dim(V) - \dim(A \cap \ker(\Enc_i))$,
    \begin{align*}
        \Phi_{V}(V, (V_1,\dots, V_n), \tau(r))
        & = \tau(r)\dim(V) - \frac{1}{n}\sum_{i=1}^n (\dim(V) - \dim(V_i))\\
        & = \tau(r)\dim(A) - \frac{1}{n}\sum_{i=1}^n \dim(A \cap \ker(\Enc_i)) < 0 \,.
    \end{align*}
    Therefore, $(V_1,\dots, V_n)$ is a $V$-local profile contained in $\calC$ with $\Phi_{V}(V, (V_1,\dots,V_n), \tau(r)) < 0$.
\end{proof}

By the ``if'' direction of \Cref{thm: local-ish for subspace-design}, we know that
if $\calC$ is \emph{not} a $\tau$-subspace design code, then $\calC$ contains a local profile $(V_1,\dots,V_n)$ with $\Phi(V, (V_1,\dots,V_n),\tau(r)) < 0$.
Using the fact that containment of a local profile is preserved under taking quotients (\Cref{lem: quotienting}), we can show a stronger statement:
if $\calC$ is \emph{not} a $\tau$-subspace design code, then $\calC$ contains a local profile $(V_1,\dots,V_n)$ with $\Phi(U, (V_1,\dots,V_n),\tau(r)) < 0$ for \emph{all} subspaces $\{0\} \neq U \subseteq V$.
This makes it easier to obtain a contradiction in our analysis in \Cref{sec:construction}.

\begin{lemma}\label{lem: too many constraints for all subspaces}
    Suppose an additive code $\calC=\{(\Enc_1(x),\dots, \Enc_n(x))\mid x\in \FF_q^k\}\subseteq \left(\FF_q^s\right)^n$ is \textbf{not} a $\tau$-subspace design. Then there exists a positive integer $r$, a linear space $V'$ of dimension at most $r$, along with a $V'$-local profile  $(V_1',\dots, V_n') \subseteq \mathcal L(V')^n$ contained in $\calC$, such that for all $\{0\}\neq U\subseteq V'$, we have $\Phi_{V'}(U, (V_1',\dots,V_n'), \tau(r))<0$.
\end{lemma}
\begin{proof}
    By \Cref{thm: local-ish for subspace-design}, if the code $\calC$ is not a subspace design, then there exists a positive integer $r$, a vector space $V$ of dimension at most $r$ and a $V$-local profile $(V_1,\dots, V_n)$ contained in $\calC$ such that 
    \[\Phi_{V}(V, (V_1,\dots, V_n), \tau(r))<0 \,.\]

    Let $W \subseteq V$ be an element of $\argmax \Phi_{V}(\ast, (V_1,\dots, V_n), \tau(r))$ of maximal dimension.
    Note that $W \neq V$ since $\Phi_{V}(V, (V_1, \dots, V_n), \tau(r))<0$, but we know that the maximum value must be non-negative since $\Phi_V(\{0\},(V_1,\dots, V_n), \tau(r))=0$. 

    Now, by \Cref{lem: quotienting}, $\calC$ also contains the $V/W$-local profile $\left((V_1+W)/W,\dots, (V_n+W)/W\right)$.
    Let $V' \coloneqq V/W$ and $V'_i \coloneqq (V_i+W)/W$ for each $i\in[n]$.
    Let $M : V \to V/W$ be the linear map that maps $v$ to the coset $v+W$.
    Then, we have $\ker(M) = W$ and $M(V_i) = V_i'$.
    
    For any non-trivial subspace $U \subseteq V$ such that $U \not\subseteq W$, let $U' = M(U) = (U+W)/W \neq \{0\}$.
    Invoking \Cref{thm: quotienting is good} with $U$, $W=\ker(M)$, $V_1',\dots,V_n'$, and $M$, we get
    \begin{align*}
        &\Phi_{V'} \left(U', \left( V_i',\dots, V_n' \right), \tau(r)\right) \\
        & \qquad =\Phi_V(U+W, (V_1,\dots, V_n), \tau(r)) - \Phi_V(W, (V_1,\dots, V_n), \tau(r)) < 0 \,,
    \end{align*}
    where the strict inequality follows since $\dim(U+W) > \dim(W)$ and $W$ is a maximal element of $\argmax \Phi_{V}(\ast, (V_1,\dots, V_n), \tau(r))$.
    The above inequality holds for all $\{0\} \neq U' \subseteq V'$, which completes the proof.
\end{proof}

In the following, we give another characterization of subspace design codes.
Here, we have a code with message space $\F_q^k$, and we consider an auxiliary space $\F_q^{k'}$ together with a linear map $M : \F_q^{k'} \to \F_q^k$.
We are interested in $M(A)$ for subspaces $A \subseteq \F_q^{k'}$.
Looking ahead, in the AEL construction, this lemma will be applied to the inner code $\Cin$ with $M$ being the encoding map of the outer code restricted to a coordinate (where $k' = k_{\mathrm{out}}$ and $k = \kin$).

\begin{lemma}\label{lem: Enc_j(A) is probably 0}
    Let $\calC=\{(\Enc_1(x),\dots, \Enc_n(x))\mid x\in \FF_q^k\}\subseteq \left(\FF_q^s\right)^n$ be an $\FF_q$-additive code which is a $\tau$-subspace design code.
    Let $V$ be a vector space and $(V_1,\dots, V_n)$ be a $V$-local profile.
    Let $A\subseteq \FF_q^{k'}$ be a subspace, $M:\FF_q^{k'}\to \FF_q^{k}$ be a linear map, and $\varphi : V\to A^\ast$ be an isomorphism satisfying $\varphi(V_i)^\circ\subseteq \ker(\Enc_i \circ M)$ for each $i\in [n]$.
    Then, either $M(A) = \{0\}$, or for any positive integer $r\ge\dim(V)$, there exists a subspace $\{0\}\neq U\subseteq V$ such that
    \[\Phi_V(U, (V_1,\dots, V_n), \tau(r))\ge 0 \,. \]
\end{lemma}
\begin{proof}
    Let $B = M(A) \subseteq \FF_q^{k}$, and assume that $B\ne \{0\}$.
    We will consider the restriction of $M$ to $A$, which we denote as $M_A: A \to B$.
    Note that we have $\dim(A) = \dim(B) + \dim(\ker(M_A))$, thus the annihilator $\ker(M_A)^\perp \subseteq A^\ast$ has $\dim(\ker(M_A)^\perp) = \dim(A) - \dim(\ker(M_A)) = \dim(B)$, which is nonzero.
    
    Let $U \coloneqq \varphi^{-1}(\ker(M_A)^\perp) \subseteq V$.
    We claim that in fact $\calC$ contains the $U$-local profile $(V_1 \cap U,\dots, V_n \cap U)$.
    Assuming this, the lemma follows by invoking \Cref{thm: local-ish for subspace-design}:
    since $U$ is a vector space of dimension $0 < \dim(U) \leq \dim(A) \leq r$, and $\calC$ contains $(V_1 \cap U,\dots, V_n \cap U)$, \Cref{thm: local-ish for subspace-design} states that
    \begin{align*}
        \Phi_V(U, (V_1 \cap U, \dots, V_n \cap U), \tau(r)) = \Phi_V(U, (V_1,\dots,V_n), \tau(r)) \geq 0 \,.
    \end{align*}
    
    To prove the claim, we first need to define an isomorphism $\varphi': U \to B^\ast$.
    For $u\in U$, we define $\varphi'(u) \in B^\ast$ as follows:
    for any $b\in B$,
    \[\varphi'(u)(b) = \varphi(u)(a) \]
    where $a \in A$ is any element satisfying $M_A a = b$.
    This is well-defined because if $a\neq a'$ have $M_A a = M_A a' = b$, then $a-a' \in \ker(M_A)$ and $\varphi(u) (a-a') = 0$ due to $\varphi(u) \in \ker(M_A)^\perp$.
    The fact that $\varphi'$ is an isomorphism is straightforward given that $\dim(U) = \dim(B) = \dim(B^\ast)$ and that $\varphi'(u) = 0 \implies \varphi(u) = 0 \implies u=0$.

    Therefore, we have a vector space $U$, a subspace $B \subseteq \F_q^k$, and an isomorphism $\varphi': U \to B^\ast$.
    To show that $\calC$ contains the $U$-local profile $(V_1\cap U,\dots, V_n \cap U)$, it suffices to show that $\varphi'(V_i \cap U)^\circ \subseteq \ker(\Enc_i)$, that is, any $b \in B$ with $\varphi'(V_i \cap U)(b) = 0$ satisfies $\Enc_i(b) = 0$.
    
    For such $b$, by definition of $\varphi'$, we have $\varphi (V_i\cap U) (a)=0$ for all $a \in A$ such that $M_A a = b$.
    Denoting $M_A^{-1} b \coloneqq \{a \in A \mid M_A a = b\}$, we have
    \begin{align*}
        M_A^{-1}b & \subseteq  \varphi(V_i \cap U)^\circ
        = \varphi(V_i)^\circ +\varphi(U)^\circ
        = \varphi(V_i)^\circ + \ker(M_A)
        \subseteq \ker(\Enc_i \circ M) + \ker(M_A) \,,
    \end{align*}
    where we use that $\varphi(U)^\circ = (\ker(M_A)^\perp)^\circ = \ker(M_A)$, and the assumption $\varphi(V_i)^\circ \subseteq \ker(\Enc_i \circ M)$.
    Therefore, we have $b \in M_A( \ker(\Enc_i \circ M) + \ker(M_A) ) \subseteq \ker(\Enc_i)$, which implies that $\Enc_i(b) = 0$ as desired.
\end{proof}

\section{Our construction}
\label{sec:construction}

We first state a simple lemma about bipartite spectral expanders.

\begin{lemma} \label{lem:EML}
    Let $G = ([n], [n],E)$ be a $d$-regular bipartite $\lambda$-expander.
    For any $x \in \R^n$, let $\mu = \frac{1}{n} \sum_{i=1}^n x_i$, and let $y \in \R^n$ be such that $y_j = \frac{1}{d} \sum_{i \in N(j)} x_i$ for each $j \in [n]$.
    Then,
    \begin{align*}
        \frac{1}{n} \sum_{j=1}^n (y_j-\mu)^2 \leq \lambda^2 \cdot \|x\|_\infty^2 \,.
    \end{align*}
\end{lemma}
\begin{proof}
    We can write $y$ as $A x$, where $A \in \R^{n \times n}$ is the normalized adjacency matrix of $G$ with $A_{i,j} = 1/d$ if $(i,j) \in E(G)$ and $0$ otherwise.
    Then,
    \begin{align*}
        \frac{1}{n} \sum_{j=1}^n \inparen{y_j - \mu}^2
        = \frac{1}{n} \norm{A x - \mu \cdot \mathbf{1} }_2^2
        = \frac{1}{n} \norm{A \inparen{x - \mu \cdot \mathbf{1}} }_2^2 \,.
    \end{align*}
    Observe that $x- \mu \cdot \mathbf{1}$ is orthogonal to $\mathbf{1}$.
    Since $G$ is a $\lambda$-expander, we have $\|Av\|_2^2 \leq \lambda^2 \|v\|_2^2$ for any vector $v \perp \mathbf{1}$.
    Thus, $\frac{1}{n} \norm{A \inparen{x - \mu \cdot \mathbf{1}} }_2^2 \leq \frac{1}{n} \lambda^2 \|x-\mu \cdot \mathbf{1}\|_2^2 \leq \lambda^2 \|x\|_\infty^2$, where we use that $\frac{1}{n} \|x-\mu \mathbf{1}\|_2^2 = \frac{1}{n} \|x\|_2^2 - \mu^2 \leq \|x\|_{\infty}^2$.
\end{proof}

\parhead{Our main theorem.}
We show that the AEL construction (\Cref{def:AEL}) inherits the subspace design property of the inner codes, as long as the outer code has constant distance and the graph is a sufficiently good expander.

\begin{theorem} \label{thm:main-AEL-theorem}
    Let $\tau: \N \to \R_{\le 1}$, $\eps > 0$, and $r \in \N$.
    Suppose
    \begin{itemize}
        \item $\Cout$ is a code with $\Eout : \F_q^k \to (\F_q^{\kin})^n$ and has relative distance $\delta_{\out} > 0$,
        \item $\Cin$ is a code with $\Ein : \F_q^{\kin} \to (\F_q^s)^d$ and is a $\tau$-subspace design code,
        \item $G = (\Vleft,\Vright,E)$ is a $d$-regular bipartite $\lambda$-expander with $|\Vleft| = |\Vright| = n$ and suppose $\lambda < \eps q^{-r^2/2} \sqrt{\delta_{\out}}$.
    \end{itemize}
    Then, $\calC = \CAEL(\Cout, \Cin, G) \subseteq ((\F_q^s)^d)^n$ a $\wt{\tau}$-subspace design code where $\wt{\tau}(r') = \tau(r')+\eps$ for all $r' \leq r$, and $\wt{\tau}(r') = 1$ otherwise.
\end{theorem}

\begin{proof}
    We first set a few notations.
    The vertices $i\in \Vright$ and $j\in \Vleft$ are naturally associated with the set $[n]$.
    We write $\Eoutj{j}: \F_q^k \to \F_q^{\kin}$ to denote the outer encoder restricted to a coordinate $j\in \Vleft$.
    Similarly, $\Einl{\ell}: \F_q^{\kin} \to \F_q^s$ for the inner code restricted to $\ell \in [d]$, and
    $\EAELi{i}: \F_q^k \to (\F_q^s)^d$ for the final encoder restricted to $i \in \Vright$.
    
    For contradiction, suppose $\calC$ is \emph{not} a $\wt{\tau}$-subspace design code.
    For simplicity, we denote $\tau' = \wt{\tau}(r) = \tau(r) +\eps$.
    By \Cref{lem: too many constraints for all subspaces}, there exist a linear space $V$ of dimension at most $r$ and subspaces $V_1,\dots,V_n \subseteq V$, such that $\calC$ contains the $V$-local profile $(V_1,\dots,V_n)$, and
    \begin{align}
        \Phi_V(U, (V_1,\dots,V_n), \tau')
        = \frac{1}{n} \sum_{i=1}^n \dim (U \cap V_i) - (1-\tau') \dim(U)
        < 0 \,,
        \label{eq:bound-on-Phi}
    \end{align}
    for all subspaces $\{0\} \neq U \subseteq V$.

    Recall that $\calC$ containing $(V_1,\dots,V_n)$ means that there is a subspace $A \subseteq \F_q^k$ (in the message space) and an isomorphism $\varphi: V \to A^*$ (where $A^*$ is the dual space of $A$) such that $\varphi(V_i)^{\circ} \subseteq \ker(\EAELi{i})$ for all $i\in[n]$.
    In other words, any vector $a \in A$ with $\varphi(V_i) (a) = 0$ satisfies $\EAELi{i}(a) = 0$.

    Recall that a codeword can be viewed as a regrouping of $\Ein(\Eoutj{j}(a)) \in (\F_q^s)^d$ for $j\in \Vleft$.
    For a vertex $j\in \Vleft$ and its $d$ neighbors $N_1(j), \dots, N_d(j) \in \Vright$, we have that $\EAELi{N_\ell(j)}(a) = 0$ implies that $\Einl{\ell}(\Eoutj{j}(a)) = 0$.
    Therefore, $\varphi(V_{N_\ell(j)})^{\circ} \subseteq \ker(\Einl{\ell} \circ \Eoutj{j})$ for each neighbor indexed by $\ell \in [d]$.

    Next, we use the $\tau$-subspace design property of $\Cin$.
    Invoking \Cref{lem: Enc_j(A) is probably 0} with $k' = k_{\out}$, $k=\kin$, the linear map $M = \Eoutj{j} : \F_q^{k_{\out}} \to \F_q^{\kin}$, and local profile $(V_{N_1(j)},\dots,V_{N_d(j)})$, it follows that either
    \begin{enumerate}[(1)]
        \item $\Eoutj{j}(A) = 0$, or 
        \item for $r = \dim(V)$, there exists a non-trivial subspace $U \subseteq V$ such that
        \begin{align*}
            \Phi \inparen{U, (V_{N_1(j)},\dots,V_{N_d(j)}), \tau(r)} \geq 0 \,.
        \end{align*}
    \end{enumerate}
    The rest of the proof goes by showing that more than $1-\delta_{\out}$ fraction of $j\in \Vleft$ has $\Eoutj{j}(A) = 0$, hence contradicting the distance $\delta_{\out}$ of the outer code $\calC_{\out}$.
    To do so, we will show that $\frac{1}{d} \sum_{\ell=1}^d \dim(U \cap V_{N_{\ell}(j)})$ is close to $\frac{1}{n} \sum_{i=1}^n \dim(U \cap V_i)$ for most $j$ and most subspaces $U$.

    Fix a non-trivial subspace $U \subseteq V$.
    Define vectors $x, y \in \R^n$ to be such that $x_i \coloneqq \dim(U \cap V_i)$ and $y_j \coloneqq \frac{1}{d} \sum_{\ell=1}^d \dim(U \cap V_{N_{\ell}(j)}) 
        = \frac{1}{d} \sum_{\ell=1}^d x_{N_{\ell}(j)}$.
    Moreover, let $\mu \coloneqq \frac{1}{n} \sum_{i=1}^n x_i$,
    and note that $x_i \leq \dim(U)$.
    Then, by \Cref{lem:EML}, we have
    \begin{align*}
        \frac{1}{n} \sum_{j=1}^n (y_j-\mu)^2 \leq \lambda^2 \cdot \dim(U)^2 \,.
    \end{align*}
    By Markov's inequality, we have that
    \begin{align*}
        \Pr_{j\sim [n]} \insquare{|y_j - \mu| \geq \eps \cdot \dim(U)} \leq \frac{\lambda^2}{\eps^2} \,.
    \end{align*}

    On the other hand, we know from \Cref{eq:bound-on-Phi} that $\mu < (1-\tau') \dim(U)$.
    Thus, if $|y_j-\mu| \leq \eps \cdot \dim(U)$, then $y_j < (1-\tau' + \eps) \dim(U)$.
    Since $\tau' = \tau(r)+\eps$, we have
    \begin{align*}
        \Phi \inparen{U, (V_{N_1(j)},\dots,V_{N_d(j)}), \tau(r)}
        = y_j - (1-\tau(r)) \dim(U) < 0 \,.
    \end{align*}
    Taking a union bound over $\leq q^{r^2}$ many non-trivial subspaces $U \subseteq V$, we have that $1- \frac{\lambda^2}{\eps^2} q^{r^2}$ fraction of $j\in \Vleft$ satisfies $\Phi(U, (V_{N_1(j)},\dots,V_{N_d(j)}), \tau(r)) < 0$ for all non-trivial subspaces $U \subseteq V$.
    For such $j\in \Vleft$, we must have $\Eoutj{j}(A) = 0$.

    Suppose $\lambda < \eps q^{-r^2/2} \sqrt{\delta_{\out}}$, then $>1-\delta_{\out}$ fraction of $j\in \Vleft$ has $\Eoutj{j}(A) = 0$, which contradicts that $\calC_{\out}$ has distance $\delta_{\out}$.
    This completes the proof.
\end{proof}

From \Cref{thm:main-AEL-theorem}, \Cref{thm:subspace-design} follows almost immediately.

\begin{proof}[Proof of \Cref{thm:subspace-design}]
    Given $r\in \N$ and $R,\eps \in (0,1)$,
    we instantiate the ingredients in \Cref{thm:main-AEL-theorem} as follows:
    \begin{itemize}
        \item Outer code $\Cout$:
        We choose an explicit additive code with rate $R_{\out} = 1-\eps/4$ and relative distance $\delta_{\out} = \eps/8$.
        
        \item Inner code $\Cin$:
        We set $s = \Theta(r/\eps)$.
        By \Cref{thm: Folded RLCs are subspace-designs}, there exists a $\tau_{\mathrm{in}}$-subspace design code in $(\F_q^s)^d$ of rate $R_{\mathrm{in}}$ with $\tau_{\mathrm{in}}(r') \leq R_{\mathrm{in}} + \eps/4$ for all $r' \leq \eps s/16 \leq r$.
        Since $s,d$ are constants, we can find such a code by brute force.

        \item Graph $G$:
        We use the explicit construction from \Cref{thm:expander-construction}, where $\lambda \leq 2/\sqrt{d}$ (as long as $n \geq n_0(d)$ is sufficiently large).
        We choose $d = \frac{\Theta(1)}{\eps^2 \delta_{\out}} q^{r^2} = \poly(1/\eps) \cdot q^{r^2}$.
    \end{itemize}
    By choosing $R_{\mathrm{in}}$ to be $R+\eps/2$, we get that the rate of $\CAEL$ is $R_{\out} R_{\mathrm{in}} \geq R$.
    The alphabet is $\F_q^{sd} = \F_q^{\poly(r,1/\eps) q^{r^2}}$.
    By \Cref{thm:main-AEL-theorem} (with $\tau_{\mathrm{in}}$ and $\eps/4$), our code is a $\tau$-subspace design code with $\tau(r') = \tau_{\mathrm{in}}(r')+\eps/4 \leq R_{\mathrm{in}} + \eps/2 \leq R+\eps$ for all $r' \leq r$.
    This completes the proof of \Cref{thm:main-AEL-theorem}.
\end{proof}

\section{Consequences for list-decoding, recovery and curve-decoding}
\label{sec:consequences}

In this section, we state the list-decoding, list-recovery, and curve-decoding parameters of our construction.

\begin{definition}
    Let $C_{q, R, r, \epsilon, n}$ be our explicit $\FF_q$-additive code of rate $R$ over alphabet $\Sigma = \FF_q^{\poly(r,1/\epsilon)\cdot q^{r^2}}$ which is a $\tau$-subspace design with $\tau(r')\le R+\epsilon$ for all $r'\le r$.
\end{definition}

We note that our constructions exist for each choice of parameters by \Cref{thm:subspace-design}.
Our construction has not been explicitly defined for any of the below properties, but it naturally captures them due to the strength of subspace designs.

\begin{theorem}[Average radius list-decoding]
    For any finite field $\FF_q$, $R, \epsilon \in (0,1)$ and positive integer $L$, the $\F_q$-additive code $C_{q, R, L, \epsilon, n} \subseteq \Sigma^n$, of rate $R$ and alphabet $\Sigma = \FF_q^{\poly(L/\epsilon)\cdot q^{L^2}}$, has the property that for any distinct $c_0,\dots, c_L \in C_{q, R, L, \epsilon, n}$ and $y\in \Sigma^n$, we have \[\sum_{i=0}^L \Delta(y, c_i)\ge L(1-R-\epsilon) \ . \]
\end{theorem}
\begin{proof}
    $C_{q, R, L, \epsilon, n} \subseteq \Sigma^n$ is a $\tau$-subspace design code with $\tau(L)\le R+\epsilon$. Now, we can apply \Cref{thm: list-decoding}.
\end{proof}

\begin{theorem}[List-recovery]
    For any finite field $\FF_q$, $R,\epsilon_0 ,\epsilon_1 \in (0,1)$ and positive integer $\ell$, the $\F_q$-additive code $C_{q, R, \lceil\ell/\epsilon_1\rceil, \epsilon_0, n} \subseteq \Sigma^n$, of rate $R$ and alphabet $\Sigma = \FF_q^{\poly(\ell, 1/\epsilon_1, 1/\epsilon_0)\cdot q^{\lceil\ell/\epsilon_1\rceil^2}}$,  has the following list-recovery property: for any $L_1,\dots, L_n\subseteq \Sigma$ with $|L_i|=\ell$ for each $i\in [n]$,  \[\left|\left\{c\in \calC \mid \Delta(c_i,L_1\times L_2\times \cdots\times L_n)<1-R-\epsilon_0-\epsilon_1 \right\}\right|\le \left(\frac{\ell}{R+\epsilon_0+\epsilon_1}\right)^{(R+\epsilon_0+\epsilon_1)/\epsilon_1} \ . \]
\end{theorem}
\begin{proof}
    $C_{q, R, \lceil\ell/\epsilon_1\rceil, \epsilon_0, n} \subseteq \Sigma^n$ is a $\tau$-subspace design code with $\tau(\lceil \ell/\epsilon_1\rceil)\le R+\epsilon_0$. Now, we can apply \Cref{thm: list-recovery}. Note that this is possible since the code $C$ we construct is thus also a $\tau'$subspace design with $\tau'(r)=\max(R+\eps_0, \tau(r))$ for all $r$.
\end{proof}

\begin{corollary}
    For any finite field $\FF_q$, $R,\epsilon \in (0,1)$ and positive integer $\ell\ge 2$, there exists a choice of $\epsilon_0$ and $\epsilon_1 \in (0,1)$ with $\epsilon_0 + \epsilon_1 = \epsilon$, such that the $\F_q$-additive code $C_{q, R, \lceil\ell/\epsilon_1\rceil, \epsilon_0, n} \subseteq \Sigma^n$ of rate $R$ and alphabet $\Sigma = \FF_q^{\poly(\ell, 1/\epsilon)\cdot q^{4\ell^2/\epsilon^2}}$ has the following list-recovery property:  for any $L_1,\dots, L_n\subseteq \Sigma$ with $|L_i|=\ell$ for each $i\in [n]$,  \[\left|\left\{c\in \calC \mid \Delta(c_i,L_1\times L_2\times \cdots\times L_n)<1-R-\epsilon\right\}\right|\le \left\lceil\left(\frac{\ell}{R+\epsilon}\right)^{(R+\epsilon)/\epsilon}\right\rceil \ . \]
\end{corollary}
\begin{proof}
    Let $L =\left\lceil\left(\frac{\ell}{R+\epsilon}\right)^{(R+\epsilon)/\epsilon}\right\rceil$. Consider $\epsilon_0 = \frac{\epsilon^2}{4L\log(\ell/\epsilon)}$ and $\epsilon_1 = \epsilon - \epsilon_0$.

    Thus, $C_{q, R, \lceil\ell/\epsilon_1\rceil, \epsilon_0, n} \subseteq \Sigma^n$, is an $\FF_q$-additive code of rate $R$, alphabet $\Sigma = \FF_q^{\poly(\ell, 1/\epsilon_1, 1/\epsilon_0)\cdot q^{\lceil\ell/\epsilon\rceil^2}} = \FF_q^{\poly(\ell, 1/\epsilon)\cdot \poly(1/\epsilon_0)\cdot q^{\lceil\ell/\epsilon\rceil^2}}$ and 
    \[ \left|\left\{c\in \calC \mid \Delta(c_i,L_1\times L_2\times \cdots\times L_n)<1-R-\epsilon\right\}\right|\le \left\lfloor\left(\frac{\ell}{R+\epsilon}\right)^{(R+\epsilon)/\epsilon_1}\right\rfloor . \]
Therefore, it suffices to prove that:
    \begin{itemize}
        \item $\poly(1/\epsilon_0)\cdot q^{\lceil\ell/\epsilon\rceil^2}\le \poly(\ell/\epsilon)\cdot q^{4\ell^2/\epsilon^2}$; and
        \item $\left(\frac{\ell}{R+\epsilon}\right)^{(R+\epsilon)/\epsilon_1}\leq  \left(\frac{\ell}{R+\epsilon}\right)^{(R+\epsilon)/\epsilon}+1/2$ \ .
    \end{itemize}

    For the first, observe that it suffices to prove that $\log L + \lceil\ell/\epsilon\rceil^2 \le 4\ell^2/\epsilon^2$ but $\log L  =\frac{R+\epsilon}{\epsilon}\cdot \log_q(\ell/(R+\epsilon))\le \ell/\eps\cdot \ell/\epsilon\le \ell^2\epsilon^2$ and thus, this part follows.

    Now, for the second part we have that: 
    \begin{align*}
        \left(\frac{\ell}{R+\epsilon}\right)^{(R+\epsilon)/\epsilon_1}& = \left(\frac{\ell}{R+\epsilon}\right)^{(R+\epsilon)(1/\epsilon+1/\epsilon_1-1/\epsilon)}\\
         & = \left(\frac{\ell}{R+\epsilon}\right)^{(R+\epsilon)/\epsilon}\cdot \left(\frac{\ell}{R+\epsilon}\right)^{(R+\epsilon)(\epsilon_0/\epsilon\cdot \epsilon_1)} \quad (\epsilon_1 \ge \epsilon/2)\\
         &\leq \left(\frac{\ell}{R+\epsilon}\right)^{(R+\epsilon)/\epsilon}\cdot (\ell/\epsilon)^{1/2L\log (\ell/\epsilon)}\\
         & \leq \left(\frac{\ell}{R+\epsilon}\right)^{(R+\epsilon)/\epsilon}\cdot 2^{1/2L}\\
         & \leq\left(\frac{\ell}{R+\epsilon}\right)^{(R+\epsilon)/\epsilon}(1+1/2L)\\
         &\leq \left(\frac{\ell}{R+\epsilon}\right)^{(R+\epsilon)/\epsilon}+1/2
    \end{align*}
    as desired.
\end{proof}
\begin{theorem}[Curve-decoding]
\label{thm: curve-decoding-main}
    For any finite field $\FF_q$, $R, \epsilon_0, \epsilon_1 \in (0,1)$ and positive integers $\ell\ge 2, a$, the $\F_q$-additive code $C_{q, R, \lceil(\ell+1)/\epsilon_1\rceil, \epsilon_0, n} \subseteq \Sigma^n$ of rate $R$, and alphabet $\Sigma = \FF_q^{\poly(\ell, 1/\epsilon_0,1/\epsilon_1)\cdot q^{{\lceil(\ell+1)/\epsilon_1\rceil}^2}}$, is $(\ell, 1-R-\epsilon_0-\epsilon_1, a, a\cdot \frac{\epsilon_1}{\lceil(\ell+1)/\epsilon_1\rceil+\epsilon_1}) $ curve-decodable.
\end{theorem}
\begin{proof}
    $C_{q, R, \lceil (\ell+1)/\epsilon_1\rceil, \epsilon_0, n} \subseteq \Sigma^n$ is a $\tau$-subspace design code with $\tau(\lceil (\ell+1)/\epsilon_1\rceil)\le R+\epsilon_0$. Now, we can apply \Cref{thm: curve-decoding}.
\end{proof}

By setting $\eps_0 = \eps_1 = \eps/2$ in \Cref{thm: curve-decoding-main}, we get the following corollary.

\begin{corollary}
    For any finite field $\FF_q$, $R, \epsilon \in (0,1)$ and positive integers $\ell\ge 2, a$, the $\F_q$-additive code $C_{q, R, 2\lceil(\ell+1)/\epsilon\rceil, \epsilon/2, n} \subseteq \Sigma^n$ of rate $R$, and alphabet $\Sigma = \FF_q^{\poly(\ell, 1/\epsilon)\cdot q^{4{\lceil(\ell+1)/\epsilon\rceil}^2}}$, is $(\ell, 1-R-\epsilon, a, \Omega(a\cdot \eps^2/\ell)) $ curve-decodable.
\end{corollary}

\section{Acknowledgments}

Part of this done was done while R.G. was visiting V.G. at the Simons Institute for the Theory of Computing. We thank Josh Brakensiek, Yeyuan Chen and Zihan Zhang for various discussions related to the recent progress on local properties of RLCs, and Zihan in particular for a chat about AEL and subspace designs.
We also thank Nikhil Shagrithaya for comments on the paper and discussions about \cite{jeronimo2025AEL}.

V.G. is supported by a Simons Investigator award, NSF grant CCF-2211972, and ONR grant N00014-24-1-2491.

R.G. is supported by (Yael Tauman Kalai’s) grant from  Defense Advanced Research Projects Agency
(DARPA) under Contract No. HR0011-25-C-0300. Any opinions, findings and conclusions or recommendations expressed in this material are those of the author(s) and do not necessarily reflect the views of the Defense Advanced Research Projects Agency (DARPA).

\addcontentsline{toc}{section}{References}
{\small \bibliographystyle{alpha}
   \bibliography{subspace-design}
}
\end{document}